\documentclass[sigconf,twocolumn]{acmart}

\DeclareMathOperator{\denom}{denom}
\newcommand{\Fp}{\mathbb{F}_{p}}
\newcommand{\lc}{\mathrm{lc}}
\newcommand{\LC}{\mathrm{LC}}
\newcommand{\lcs}{\mathfrak{lc}}
\newcommand{\LCs}{\mathfrak{lc}}
\newcommand{\tc}{\mathrm{tc}}
\newcommand{\TC}{\mathfrak{tc}}
\newcommand{\ord}{\mathop{\rm ord}}

\newtheorem{remark}{Remark}
\usepackage[ruled,linesnumbered]{algorithm2e}
\usepackage{multicol,multirow}

\AtBeginDocument{%
  \providecommand\BibTeX{{%
    \normalfont B\kern-0.5em{\scshape i\kern-0.25em b}\kern-0.8em\TeX}}}

\setcopyright{acmcopyright}
\copyrightyear{2022}
\acmYear{2022}
\acmDOI{XXXXXXX.XXXXXXX}

\acmConference[ISSAC 2022]{International Symposium on Symbolic and Algebraic Computation}{July 04--07,
  2022}{Lille, France}
\acmPrice{15.00}
\acmISBN{978-1-4503-XXXX-X/18/06}

\begin{document}

%%
%% The "title" command has an optional parameter,
%% allowing the author to define a "short title" to be used in page headers.

% \title{Desingularization and p-Characteristic Polynomial of Linear Recurrence Operators}
%
% Decide which is best:
%
\title{Desingularization and p-Curvature of Recurrence Operators}

%%
%% The "author" command and its associated commands are used to define
%% the authors and their affiliations.
%% Of note is the shared affiliation of the first two authors, and the
%% "authornote" and "authornotemark" commands
%% used to denote shared contribution to the research.

\author{Yi Zhou}
\affiliation{%
  \institution{Florida State University}
  \city{Tallahassee}
  \country{USA}}
\email{yzhou@math.fsu.edu}

\author{Mark van Hoeij}
\affiliation{%
  \institution{Florida State University}
  \city{Tallahassee}
  \country{USA}}
\email{hoeij@math.fsu.edu}

%%
%% By default, the full list of authors will be used in the page
%% headers. Often, this list is too long, and will overlap
%% other information printed in the page headers. This command allows
%% the author to define a more concise list
%% of authors' names for this purpose.
\renewcommand{\shortauthors}{Zhou and van Hoeij}

%%
%% The abstract is a short summary of the work to be presented in the
%% article.
\begin{abstract}
Linear recurrence operators in characteristic $p$ are classified by their $p$-curvature. For a recurrence operator $L$, denote by $\chi(L)$ the characteristic polynomial of its $p$-curvature.
We can obtain information about the factorization of $L$ by factoring $\chi(L)$. The main theorem of this paper gives an unexpected relation between $\chi(L)$ and the true singularities of $L$.
An application is to speed up a fast algorithm for computing $\chi(L)$ by desingularizing $L$ first. Another contribution of this paper is faster desingularization.
\end{abstract}

%%
%% The code below is generated by the tool at http://dl.acm.org/ccs.cfm.
%% Please copy and paste the code instead of the example below.
%%


%%
%% Keywords. The author(s) should pick words that accurately describe
%% the work being presented. Separate the keywords with commas.

%% A "teaser" image appears between the author and affiliation
%% information and the body of the document, and typically spans the
%% page.
%\begin{teaserfigure}
%  \includegraphics[width=\textwidth]{sampleteaser}
%  \caption{Seattle Mariners at Spring Training, 2010.}
%  \Description{Enjoying the baseball game from the third-base
%  seats. Ichiro Suzuki preparing to bat.}
%  \label{fig:teaser}
%\end{teaserfigure}

%%
%% This command processes the author and affiliation and title
%% information and builds the first part of the formatted document.
\maketitle

\section{Introduction}
Singularities of linear difference operators can be divided into two groups, true (i.e. non-removable) singularities,
and apparant (i.e. removable) singularities. Desingularization (detecting or removing apparant singularities)
can expedite various algorithms for difference or differential equations.
An early application \cite{HomURL} appeared in  \verb+DEtools[Homomorphisms]+
in Maple 10. Other algorithms that benefit from reducing the number of singularities
include finding closed form solutions
%  \cite{chalevyvanhoeij},  --> not explicitly mentioned there.  Instead I'll move the other examples here
and factoring, e.g. \verb+LREtools[RightFactors]+ in Maple 2021.

In characteristic $p$, difference operators can be classified by the so-called $p$-curvature.  Our main
result gives a relation between $\chi(L)$, the characteristic polynomial of the $p$-curvature of $L$, and the true singularities of $L$. We
prove that the denominator of $\chi(L)$ determines the true singularities, including their multiplicities, up to shift equivalence.
%equals the norm of ${\rm LC}(L)$
%(a polynomial that represents the true singularities of $L$).

The algorithm from \cite{Bostan:2014:FAC:2608628.2608650} computes $\chi(L)$, multiplied by a denominator bound,
by computing its $Z$-adic expansion.
One application of our theorem is that we can replace the denominator bound by the exact denominator.
This lowers the required $Z$-adic precision, which can speed up the computation, see \autoref{sec:pcurv_timing}.
% and as already mentioned, other algorithms benefit from desingularization as well. 

We want desingularization to take less time than the time it saves in applications. Then it is useful to compute a
{\em partial desingularization} (where the goal is to remove most apparant singularities, at a fraction of the cost of a full desingularization).
We give various algorithms for this in \autoref{sec:desing1}.

Both authors were supported by NSF grants 1618657 and 2007959.

\section{Preliminaries}
\subsection{Desingularization}
Let $F$ be a field. Let $P = F[x][y]$ and $D = F(x)[y]$. If $f \in P$, then $f$ is called {\em primitive} if the gcd of its coefficients in $F[x]$ is 1.
If $f \in D - \{0\}$, then there is $c \in F(x)-\{0\}$, unique up to a factor in $F$, for which $c^{-1} f \in P$ is primitive. The content of $f$,
denoted ${\rm Cont}(f)$, is this $c$,
while the {\em primitive part} of $f$ is ${\rm Prim}(f) = c^{-1} f \in P$.
A version of Gauss's lemma says ${\rm Cont}(f_1 f_2) = {\rm Cont}(f_1) {\rm Cont}(f_2)$.

Let $\tau$ be the shift-operator. If $r(x)$ is a rational function then the product $\tau \cdot r(x)$ equals $r(x+1) \cdot \tau$. This product turns
$\mathcal{P} := F[x][\tau]$ and $\mathcal{D} := F(x)[\tau]$ into non-commutative rings. The product corresponds to compositions of operators, where $L=\sum_{i=0}^n a_i(x)\tau^i\in\mathcal{D}$ operates on $y(x)$ as $L(y(x))=\sum_{i=0}^n a_i(x)y(x+i)$.

If $L \in \mathcal{D}-\{0\}$ we can define ${\rm Cont}(L) \in F(x)$ and ${\rm Prim}(L) \in \mathcal{P}$ in the same way as before.

\begin{definition} \label{Gaussian} An operator $L \in \mathcal{P}$ is called {\em Gaussian}
if  
$$\forall_{A \in \mathcal{D}} \ A L \in \mathcal{P} \Longrightarrow A \in \mathcal{P}.$$
\end{definition}

If $L$ is Gaussian then $L$ is primitive. In the commutative case, the two properties are equivalent by Gauss's lemma.
It is known that Gauss's lemma does not hold in
the non-commutative case, which is illustrated in Example~\ref{ex:desing0} below (see also \cite[p.~27]{Jaroschek}).

% The reason for this definition is as follows.
An element $L = \sum_{i=0}^n a_i(x) \tau^i \in \mathcal{P}$ corresponds to a recurrence relation $L(y(x))=0$, i.e.
\begin{equation}\label{eq:rec}
 a_n(x)y(x+n)+a_{n-1}(x)y(x+n-1)+\cdots +a_0(x)y(x)=0.
\end{equation}
So we can express $y(r)$ in terms of $y(r-1),\ldots,y(r-n)$ where $r = x+n$.
The expression is not defined when $r$ is a root of the denominator, which is $a_n(r-n)$. Hence we define:

\begin{definition}
Let $L = \sum_{i=0}^n a_i(x) \tau^i \in \mathcal{P}$ be primitive. If $a_n \neq 0$ then define ${\rm ord}(L) := n$ and let $\lcs(L)$ be $a_n(x-n)$ divided
by its leading coefficient (to make it monic).
%
% Move this line elsewhere or just delete?
% (This $\lcs(L)$ appears as leading coefficient if we rewrite $L$ in the non-standard form $\sum_{i=0}^n \tau^i a_i(x-i)$.)
%
The \emph{singularities} of $L$ are the monic irreducible factors of $\lcs(L)$ in $F[x]$ (or equivalently, their roots in $\overline{F}$) with their multiplicities.
% We write these roots (or factors) with their multiplicities.
\end{definition}

\begin{example}\label{ex:desing0}
Let
$$L=x^2(x^2+1)\tau-(x+1)(x^2+2x+2)\in\mathbb{Q}[x][\tau].$$
Substituting $x \mapsto x-{\rm ord}(L)$ in the leading coefficient we obtain
$\lcs(L) = (x-1)^2 ( (x-1)^2+1 )$.
With repetition indicating multiplicity,
the singularities are $x-1$, $x-1$, $(x-1)^2+1$, or equivalently $1$, \ $1$, \ $1 \pm \sqrt{-1}$.
Let
$$A={\frac {1}{ \left( x+1 \right)  \left( {x}
^{2}+2\,x+2 \right) }}(10\,\tau+11\,{x}^{2}+15\,x+14).$$
Then
$$AL=10 \left( x+1 \right) {\tau}^{2}+ \left( 11\,{x}^{3}-18\,{x}^{2}+35
\,x-50 \right) \tau-11\,{x}^{2}-15\,x-14.
$$
Now $\lcs(AL) = x+1-{\rm ord}(AL) = x-1$ (we divided by 10 to make it monic).
Compared to $L$,
        % , the list of singularities reduces from $x-1$, $x-1$, $(x-1)^2+1$ to just $x-1$, so this \emph{desingularization}
the singularity $(x-1)^2+1$ disappeared, as well as one of the two copies of $x-1$.
\end{example}

\begin{lemma}[{\cite[Theorem~4.1.7, Corollary~4.1.9]{Jaroschek}} ]\label{lemma:LC} % \todo{Search in the literature to see if this is proved somewhere else}
Let $L\in \mathcal{P}$. For $k=0,1,2,\ldots$ let
    \begin{equation}\label{eq:lcideal}
    \mathcal{I}_k= \{0\}\,\cup\,\{\lcs(AL) \ : \ A \in \mathcal{D}, AL \in \mathcal{P}, {\rm ord}(A) = k\}
\end{equation}
and let $\mathcal{I}_{\infty}$ be their union. Then
% \begin{equation}\label{eq:lcidealchain}
$\mathcal{I}_0\subseteq \mathcal{I}_1\subseteq \mathcal{I}_2 \subseteq \cdots \subseteq \mathcal{I}_{\infty}$
% \end{equation}
are ideals in $F[x]$.
\end{lemma}
\begin{proof}

To show $\mathcal{I}_k \subseteq \mathcal{I}_{k+1}$, take a non-zero $a \in \mathcal{I}_k$. So $a = \lcs(AL)$ for some $A$ of order $k$. Replacing $A$ by $\tau A$ shows that $a \in \mathcal{I}_{k+1}$.
Clearly $\mathcal{I}_k$ is closed under $F[x]$-multiplication, and it is not difficult to show that it is closed
under addition as well.
%
% to show that it is an ideal, it remains to show that $\mathcal{I}_k(L)$ is closed under non-zero additions.
%
% Suppose $a,b \in \mathcal{I}_k$ are not 0.
% Then $a = \lcs(A)$ and $b=\lcs(BL)$ for some $A,B \in \mathcal{D}$ of order $k$ with $AL, BL \in \mathcal{P}$.
% If $a+b=0$ then there is nothing to prove, so assume $a+b \neq 0$.  Then $a+b = \lcs((A+B)L)$ so $a+b \in \mathcal{I}_k$.
\end{proof}

\begin{definition}[Essential parts, removable parts, {\cite[Definitions 4.1.8 and 4.1.10]{Jaroschek}}]
\label{def:essentialpart}
With notations as in Lemma~\ref{lemma:LC},
let $\LCs_k(L)$ be the monic generator of $\mathcal{I}_k$ for $k=0,1,2,\ldots,\infty$. %, and $\LCs_\infty(L)=\LCs_\infty(L)$ the monic generator of $\mathcal{I}_{\infty}$.
        %  and $\lcs(L)$ a generator of $\mathcal{I}$. Let $\LC_k(L)=\tau^{\deg_\tau(L)}(\LCs_k(L))$ and $\LC(L)=\tau^{\deg_\tau(L)}(\lcs(L))$.
%(In \cite{Jaroschek} this is called the \emph{essential part} of the the leading coefficient.)
Call $\LCs_k(L)$ the \emph{essential part of the leading coefficient at order $k$}. 
Note that $\LCs_0(L)=\lcs(\mathrm{Prim}(L))$ and $\LCs_{l}(L)$ divides $\LCs_{k}(L)$ if $l \geq k$. Let $\mathfrak{rp}_k(L)=\frac{\lcs(L)}{\lcs_k(L)}\in F[x]$ and call it the \emph{removable part of the leading coefficient at order $k$}.
\end{definition}

Factors (or roots) of $\lcs(L)$ are divided into two (possibly overlapping)
sublists: Factors of
$\lcs_\infty(L)$ are the \emph{true singularities} of $L$. Factors of $\mathfrak{rp}_\infty(L)$ are the \emph{apparent singularities}.
In \autoref{ex:desing0} the true singularity is $x-1$ and the apparant singularities are $x-1$, $(x-1)^2-1$.

\begin{definition}
Let $L\in\mathcal{P}$ be primitive. We call $A \in \mathcal{D}$ a \emph{desingularizer} if $AL \in \mathcal{P}$.
Such $A$ is \emph{trivial} if $A \in \mathcal{P}$ (that implies $\lcs( L ) \ | \ \lcs( AL )$, so no singularities were removed).
A desingularizer $A$ is \emph{optimal at order $k$} if $\lcs(AL)=\LCs_k(L)$ and $\ord(A)\leq k$; when $k=\infty$ say $A$ is \emph{optimal}
since it removes all apparant singularities while introducing no new singularities.
\end{definition}
An optimal desingularizer at order $k$ exists because $\mathcal{I}_{k}$ is a principal ideal; $F[x]$ is a PID.

\subsection{LCLM method for desingularization}\label{section:LCLMmethod}
Desingularization of recurrence operators has been studied in \cite{abramov1999}, \cite{vanhoeij2006}, \cite{chen2013}, \cite{Jaroschek}, \cite{CHEN2016617} and \cite{Zhang2016}. 
Papers \cite{abramov1999}, \cite{vanhoeij2006}, \cite{chen2013}, \cite{Jaroschek} aim for full desingularization; \cite{Zhang2016} focuses on desingularization over $R[x]$ where $R$ is not a field. %As is mentioned in \autoref{rem:balance}, full desingularization does not fit our purpose.
Here we describe the so-called \emph{LCLM method}, % which was used in \cite{HomURL} and later
published in \cite{CHEN2016617}. LCLM stands for the least common left multiple and GCRD the greatest common right divisor. The main result of \cite{CHEN2016617} is restated
below, where we focus on the recurrence case while the original version also applies to other types of Ore operators.
\begin{theorem}[Reformulation of Theorem 6 in \cite{CHEN2016617}]\label{thm:lclmdesing}
Suppose $L\in \mathcal{P}$. Introduce new constants $c_0,c_1,\ldots, c_k$ that are algebraically independent over $F$. Denote $A=c_0+c_1\tau+\cdots +c_k\tau^k$ and $L'=\mathrm{Prim}(\mathrm{LCLM}(L,A))\in F(c_0,c_1,\ldots,c_k)[x][\tau]$. Then $\lcs(L')=\LCs_k(L)f$ where $f\in F(c_0,c_1,\ldots,c_k)[x]$ has no non-trivial factor in $F[x]$.
\end{theorem}
The original form of Theorem 6 in \cite{CHEN2016617}:\newline
\emph{
Let $q$ be an irreducible polynomial which appears with multiplicity $e$ in $\lcs(L)$ and let $m\leq e$ be maximal such that $q^m\mid \frac{\lcs(L)}{\LCs_k(L)}$ for $k\in\mathbb{N}$. Let $A=c_0+c_1\tau+\cdots +c_k\tau^k$ in $F(c_0,\ldots,c_k)[x][\tau]$, where $c_0,\ldots, c_k$ are new constants that are algebraically independent over $F$. Denote $L'=\mathrm{Prim}(\mathrm{LCLM}(L,A))$. Then the multiplicity of $q$ in $\lcs(L')$ is $e-m$.
}

The proof in \cite{CHEN2016617} also holds when $e=m=0$, which results in \autoref{thm:lclmdesing}. It was stated for the case $\mathrm{char}(F)=0$, but the proof is valid for positive characteristic as well. 
\begin{remark}\label{rem:lcfieldext}
\autoref{thm:lclmdesing} implies $\LCs_k(L)$ stays the same if $L$ is viewed as an operator in $E(x)[\tau]$ where $E$ is a field extension of $F$, as the field extension does not affect $L'=\mathrm{Prim}(\mathrm{LCLM}(L,A))$.
\end{remark}

\autoref{thm:lclmdesing} implies the following desingularization algorithm.

%Theorem 6 of \cite{CHEN2016617} states that, assuming $\mathrm{char}(F)=0$, the equation $\texttt{LCLM\_Method}(L,k)=\LC_k(L)$ holds with a high probability.
\begin{algorithm}
\SetKwInOut{Input}{Input}\SetKwInOut{Output}{Output}
\caption{\texttt{LCLM\_Method}}\label{alg:LCLM}
\Input{a primitive operator $L=\sum_{i=0}^n a_i \tau^i\in F[x][\tau]$ and positive integer $k$}
%\Output{a polynomial in $F[x]$ which is divisible by $\LC_k(L)$, and with high probability equals $\LC_k(L)$}
\Output{$\LCs_k(L)$}
\BlankLine

$A\leftarrow \sum_{i=0}^{k} c_i\tau^i$, where $c_0,c_1,\ldots, c_{k}$ are new constants that are algebraically independent over $F$;

$L'\leftarrow \mathrm{Prim}(\mathrm{LCLM}(A,L))$;

\Return{$\gcd(\lcs(L),\lcs(L')$};
\end{algorithm}

The discussion following the main theorem in \cite{CHEN2016617} states that in characteristic 0, instead of new constants, we can let $c_0,c_1,\ldots, c_k$ be random elements in $F$. In this case the algorithm is Monte-Carlo, meaning it returns the desired result with a high probability. The Monte-Carlo version is much faster since it avoids computations in a transcendental extension of $F$.
It was implemented \cite{HomURL} with $k=1$ by the second author in 2004.

%When $c_0,c_1,\ldots, c_{k}$ are random elements in $F$ but with a low probability does not return the desired result $\LC_k(L)$.
%If $c_0,c_1,\ldots, c_{k}$ are new constants the output is guaranteed to be $\LC_k(L)$ but the algorithm is slow.

We refer to the algorithm where $k=1$ as \emph{order-1 LCLM method}. The Monte-Carlo version of order-1 LCLM method is useful in practice since it strikes a good balance between benefit and cost.
The LCLM computation is much faster for $k=1$ than for larger $k$, 
and $\lcs_1(L)$ is often very close to $\LCs_\infty(L)$. We will further speed up the algorithm in Section~\ref{sec:desing1}.

\iffalse
Apart from its practical application, \autoref{thm:lclmdesing} is also of theoretical significance. We will refer to the following remarks later. 
\begin{remark}
 In \cite{CHEN2016617}, it is assumed that $\mathrm{char}(F)=0$, but their proof for \autoref{thm:lclmdesing} is in fact valid for positive characteristic as well.
\end{remark}
\begin{remark}\label{rem:lclmdesing1}
Use the same notations as \autoref{thm:lclmdesing}.
 We emphasize that $e$ can be 0. In other words, if $f$ does not appear in $\lcs(L)$ at all, then does not in $\lcs(L')$ as well. This implies $\lcs(L')=\LCs_k(L)g$, where $g\in F(c_0,\ldots,c_k)[x]$ has no non-trivial factor in $F[x]$.
\end{remark}
\begin{remark}\label{rem:lclmdesing}
 \autoref{rem:lclmdesing1} implies that $\LC_k(L)$ stays the same if $F$ is replaced with a field extension $E\supseteq F$, because the field extension does not affect $L'=\mathrm{Prim}(\mathrm{LCLM}(L,A)).$ 
\end{remark}
\fi
%Finally it is worth noting that the LCLM method also applies to trailing coefficient.

%\section{Recurrence Operators in Positive Characteristic}
%Linear recurrence operators in positive characteristics have been studied in \cite{van2003galois}\todo{chapter} and \cite{Bostan:2014:FAC:2608628.2608650}. Definitions and most results in this section can be found in the latter. 

\subsection{p-Characteristic polynomial}
From here until Section~\ref{sec:desing1}, $F$ will be a field of characteristic $p$, where $p$ is a prime number.

A general theory of linear difference equations in positive characteristic is developed in \cite[Chapter 5]{van2003galois}. In \cite{Bostan:2014:FAC:2608628.2608650}, $p$-characteristic polynomials of recurrence operators (and differential operators) over $\Fp[x]$ are studied and an algorithm for computing them is given. An algorithm for computing $p$-characteristic polynomials of operators in $\mathbb{Z}[x][\tau]$ for a number of $p$ is presented in \cite{pages}, based on the algorithm from \cite{Bostan:2014:FAC:2608628.2608650}. We will give more information about these algorithms
% in \cite{Bostan:2014:FAC:2608628.2608650} and \cite{pages}
in \autoref{sec:BCS}

 Let $Z=x^p-x=x(x+1)\cdots (x+p-1)$. Clearly $Z$ is fixed by $\tau$ and hence elements of $F(Z)$ are $\tau$-constant. In fact, $F(Z)$ \emph{is} the field of $\tau$-constants; to see this, notice that $F(x)$ is a degree $p$ field extension of $F(Z)$ and hence there is no proper intermediate field. 
  
Let $\mathcal{N}:F(x)\rightarrow F(Z)$ denote the norm map of the field extension $F(x)/F(Z)$. It is given by the formula 
$$\mathcal{N}: f(x)\mapsto f(x)f(x+1)\cdots f(x+p-1).$$

Denote $T=\tau^p$. The center of $\mathcal{D}$ is $F(Z)[T]$. Since $F(x)[T]\subseteq \mathcal{D}$, any $\mathcal{D}$-module is naturally an $F(x)[T]$-module, that is, an $F(x)$-vector space equipped with an $F(x)$-linear map.  % induced by $T$.  is mentioned in the next line:

{
\begin{definition}
For a $\mathcal{D}$-module $M$, call the $F(x)$-linear map induced by $T$ the \emph{$p$-curvature} of $M$.% and its characteristic polynomial the \emph{$p$-characteristic polynomial}. %Denote by $\chi(M)\in \Fp(x)[T]$ the $p$-characteristic polynomial.

For an operator $L\in \mathcal{D}$, define its $p$-curvature to be that of $\mathcal{D}/\mathcal{D}L$. Denote $\chi(L)\in F(x)[T]$ its characteristic polynomial with $T$ as its variable and call it the \emph{$p$-characteristic polynomial} of $L$.

{A characteristic polynomial is monic by definition so the leading coefficient of $L$ is lost in $\chi(L)$. To reinsert it,} denote $\tilde{\chi}(L)=\mathcal{N}(\lcs(L))\chi(L)$.
%Denote $\tilde{\chi}(L)=\mathcal{N}(\lc(L))\chi(L)$. % for $L\in D_p$. 
It is called the \emph{reduced norm}
%\footnote{(maybe thesis only) In fact, $\Fp(x)[\tau]$ is an Azumaya algebra, on which the so-called reduced norm is defined. It is proved in \cite{Bostan:2014:FAC:2608628.2608650} that $\tilde{\chi}$ coincides with the reduced norm map on $\Fp(x)[\tau]$ viewed as an Azumaya algebra.} 
of $L$ in \cite{Bostan:2014:FAC:2608628.2608650}.
\end{definition}
}

%We note that $F(x)[T]$ and $F(x)[\tau^p]$ are isomorphic to each other. 
%For our purpose it is convenient to identify $T$ with $\tau^p$, which means $\chi(L)\in F(x)[\tau^p]$ is also an operator. By doing so we can apply Cayley-Hamilton Theorem (see ... in lemma ...). 

%In the following lemma we present properties of $p$-characteristic polynomials that will be used.
\begin{lemma}\label{lemma:pcurvprop}
Properties of $p$-characteristic polynomials.
\begin{enumerate}
%    \item[1] If $L_1,L_2\in \mathcal{D}$ are gauge equivalent\todo{cite} then $\chi(L_1)=\chi(L_2)$.
    \item[(i)] For $L\in \mathcal{D}$, $\chi(L)\in F(Z)[T]$. 
    \item[(ii)] For $L\in \mathcal{D}$, $\chi(L)\in \mathcal{D}L$.
    \item[(iii)] For $L_1,L_2\in\mathcal{D}$, $\chi(L_1L_2)=\chi(L_1)\chi(L_2)$ and $\tilde{\chi}(L_1L_2)=\tilde{\chi}(L_1)\tilde{\chi}(L_2)$.
    \item[(iv)] For $L_1,L_2\in\mathcal{D}$, if $\,\mathrm{GCRD}(L_1,L_2)=1$, then $$\chi(\mathrm{LCLM}(L_1,L_2))=\chi(L_1)\chi(L_2).$$
    \item[(v)] For $L\in \mathcal{P}$, $\tilde{\chi}(L)\in F[Z][T]$ and $\deg_Z(\tilde{\chi}(L))\leq \deg_x(L)$.
    \item[(vi)] If $L\in F(Z)[T]$ then $\tilde{\chi}(L)=L^p$. 
\end{enumerate}
\end{lemma}
\begin{proof}
All except item (iv) are proved in \cite[Section 3]{Bostan:2014:FAC:2608628.2608650} for the case $F=\Fp$ and the proofs are valid for a general field $F$ with positive characteristic. We now prove (iv). Denote $L=\mathrm{LCLM}(L_1,L_2)$. If $\mathrm{GCRD}(L_1,L_2)=1$ then $\mathcal{D}/\mathcal{D}L\cong \mathcal{D}/\mathcal{D}L_1\oplus \mathcal{D}/\mathcal{D}L_2$ as $\mathcal{D}$-modules (and hence as $F(x)[T]$-modules). Now (iv) follows from the fact that characteristic polynomials are multiplicative on direct sums.
\end{proof}
{\color{blue}
}

\hyperref[lemma:pcurvprop]{Lemma~\ref*{lemma:pcurvprop}(iii)} implies that an operator factors only when its $p$-characteristic polynomial factors (as a polynomial in $F(Z)[T]$). In fact, the $p$-characteristic polynomial tells us even more. See \cite{cluzeaufactorization} and \cite{VANDERPUT1996367} for discussions on this topic in the differential case. The $p$-characteristic polynomial is also useful for testing or proving irreducibility of operators in $\mathbb{Q}(x)[\tau]$ by reduction modulo $p$.

\subsection{BCS algorithm and Pag\`{e}s' algorithm}\label{sec:BCS}
Bostan, Caruso and Schost (2015) present an algorithm for computing the $p$-characteristic polynomial of an operator in $\Fp[x][\tau]$, called \texttt{Xi\_theta\_d} in \cite{Bostan:2014:FAC:2608628.2608650}. We refer to it as the BCS algorithm. %We also call the algorithm BCS algorithm, where BCS refers to three authors of \cite{Bostan:2014:FAC:2608628.2608650}. 
Their implementation in Magma is available at \url{https://github.com/schost}. 

The BCS algorithm takes a prime $p$ and a difference operator $L\in \Fp[x][\tau]$ as its input and computes $\tilde{\chi}(L)\in\Fp[Z][T]$ (making $\tilde{\chi}(L)$ monic gives $\chi(L)$). The algorithm computes $\tilde{\chi}(L)$
in $\Fp[[Z]][T]$ to precision $O(Z^{\deg_x(L)+1})$
which suffices by \hyperref[lemma:pcurvprop]{Lemma~\ref*{lemma:pcurvprop}(v)}.

For $L \in \mathcal{P}$, the first part of \hyperref[lemma:pcurvprop]{Lemma~\ref{lemma:pcurvprop}(v)} implies
that $\mathcal{N}(\lcs(L))$ is a denominator bound for $\chi(L)$. The BCS algorithm uses this bound to ensure that what it computes in $\Fp[[Z]][T]$ is in $\Fp[Z][T]$, not just in $\Fp(Z)[T]$.
We will show that (partial) desingularization leads to sharper denominator bounds. That reduces the required $Z$-adic precision, speeding up the computation. In fact, our main result \autoref{lemma:denom_pcurv}
says that full desingularization gives the exact denominator.
%, shows that the true singularities of $L$ determine the exact denominator. 

	For an operator $L\in \mathbb{Q}(x)[\tau]$, denote by $\chi_p(L)$ the $p$-characteristic polynomial of its reduction modulo $p$.
Pag\`{e}s (2021) gives an algorithm for computing $\chi_p(L)$ for a number of primes at the same time, if $L\in\mathbb{Z}[x][\tau]$ has a leading coefficient in $\mathbb{Z}$ (\cite[Algorithm 3]{pages}). The algorithm is based on the BCS algorithm.

\section{Main Theorem and Corollaries}
Let $\denom(\cdot)$ be the monic denominator of a rational function, or of a polynomial with rational function coefficients. We will use this notation in the cases $F[x]\subset F(x)$ and $F[Z][T]\subset F(Z)[T]$.
%Let $\denom(f)$ be the monic denominator of $f$. We use this notation in the following cases: $f\in F[Z][T]\subseteq F(Z)[T]$; $f\in F[x]\subseteq F(x)$.
\begin{theorem}\label{lemma:denom_pcurv}
For $L\in F[x][\tau]$, %the denominator of
%$\mathrm{\chi}(L)$ equals $\mathcal{N}(\LCs_\infty(L))$
%up to a scalar in $F^*$.
$\denom(\chi(L))=\mathcal{N}(\lcs_\infty(L))$.
\end{theorem}

The theorem quickly implies two corollaries, expressed in terms of the following definition.
\begin{definition}
Let $r_1,r_2\in F(x)-\{0\}$. We say that $r_1$ and $r_2$ are \emph{shift equivalent}, denoted $r_1\sim r_2$, if $\tau-\frac{r_1}{r_2}$ has a non-zero solution in $F(x)$, in other words,
if there exists $f\in F(x)-\{0\}$ for which $\frac{r_1}{r_2}=\frac{\tau(f)}{f}$. 
\end{definition}
If  $r(x)$ has a factor $q(x)$ in the numerator or denominator, and one replaces $q(x)$ by its shift $q(x+1)$,
then the result is shift-equivalent to $r(x)$. Note that $r_1\sim r_2$ if and only if $\mathcal{N}(r_1)=\mathcal{N}(r_2)$.

\begin{corollary}
If $\mathcal{D}/\mathcal{D}L_1\cong \mathcal{D}/\mathcal{D}L_2$ for  $L_1,L_2\in \mathcal{D}$, % are gauge equivalent, 
then $\LCs_\infty(L_1)$ and $\LCs_\infty(L_2)$ are shift equivalent, so $L_1$ and $L_2$ have the same true singularities up to shifts. %$\mathcal{N}(\LC(L_1))=\mathcal{N}(\LC(L_2))$.
\end{corollary}
%\begin{proof}
%It follows from the definition that $\chi(L_1)=\chi(L_2)$. A simple application of \autoref{lemma:denom_pcurv} shows $\mathcal{N}(L_1)=\mathcal{N}(L_2)$.
%\end{proof}
\begin{corollary}
For $L_1,L_2\in \mathcal{D}$, if %$L=L_1L_2$, or $L=\mathrm{LCLM}(L_1,L_2)$ and $\mathrm{GCRD}(L_1,L_2)=1$,
\begin{itemize}
    \item $L=L_1L_2$, or
    \item $L=\mathrm{LCLM}(L_1,L_2)$ and $\mathrm{GCRD}(L_1,L_2)=1$
\end{itemize}
then $\LCs_\infty(L)$ and $\LCs_\infty(L_1)\LCs_\infty(L_2)$ are shift equivalent.
%$$\mathcal{N}(\LC(L))=\mathcal{N}(\LC(L_1)\LC(L_2)).$$
%In other words, 
%$\mathcal{N}(\LC(L))=\mathcal{N}(\LC(L_1)\LC(L_2))$.
\end{corollary}
These corollaries are not true in characteristic 0, so we did not expect Theorem~\ref{lemma:denom_pcurv}.

\section{Proof of the Main Theorem}
This section is devoted to the proof of \autoref{lemma:denom_pcurv}. We start with an easy lemma.
\begin{lemma}\label{lemma:easy} For any $A, L \in \mathcal{D}$
$$\denom(\chi(AL))=\denom(\chi(A))\denom(\chi(L)).$$
\end{lemma}
\begin{proof}
Lemma~\ref*{lemma:pcurvprop}(iii) and Gauss's lemma for $F(Z)[T]$ gives
$$\mathrm{Cont}(\chi(A L))=\mathrm{Cont}(\chi(A))\mathrm{Cont}(\chi(L)).$$
But $\chi(\cdot)$ is monic by definition, and the denominator of a monic polynomial is the reciprocal of the content.
\end{proof}

%Notations:\todo{maybe move this to an earlier place. Also need to define denominator} when we say a polynomial in $\Fp(Z)[T]$ is primitive, it means its coefficients are in $\Fp[Z]$ and have $\gcd$ 1; contents and primitive parts of polynomials in $\Fp(Z)[T]$ are also taken in this sense.

% \subsection{Gaussian Operators}
% Merge this section with the next one.
\subsection{Special case, Gaussian operators}

\begin{lemma}\label{lemma:gaussian} Let $L \in \mathcal{P}$. The following are equivalent.
\begin{enumerate}
\item[1.] $L$ is Gaussian, i.e. $\forall_{A \in \mathcal{D}} \ A L \in \mathcal{P} \Longrightarrow A \in \mathcal{P}$.
\item[2.] Every desingularizer is trivial.
\item[3.] ${\rm Cl}(L)=\mathcal{P}L$, where  ${\rm Cl}(L) := \mathcal{D}L \bigcap \mathcal{P}$. (This is called the \emph{Weyl closure} in  \cite{Tsai00weylclosure}.) %(recall that ${\rm Cl}(L) = \mathcal{D}L \bigcap \mathcal{P}$)
\item[4.] $\lcs(L)=\LCs_\infty(L)$, i.e. there are no apparant singularities.
\end{enumerate}
\end{lemma}

\begin{proof}
Items 2 and 3 are reformulations of item 1, and immediately imply item 4. %Item 4 implies item 1 is in fact a direct application of Lemma~\ref{lemma:desingularizerdenombound}.
It remains to show that item 4 implies item 1. Suppose that $\lcs(L)=\LCs_\infty(L)$ and $AL \in \mathcal{P}$.
To prove: $A \in \mathcal{P}$.

By partial fraction decomposition, $A = A_1 + A_2$ where $A_1 \in \mathcal{P}$ and $A_2 = \sum q_i \tau^i$ with the numerator of $q_i$ having lower degree than its denominator.
Since $AL$ and $A_1 L$ are in $\mathcal{P}$, their difference $A_2 L$ is in $\mathcal{P}$ as well. If $A_2 \neq 0$, then the leading coefficient of $A_2 L$ will have lower degree than $\lcs(L)$,
contradicting item~4. Thus $A_2 = 0$ and hence $A \in \mathcal{P}$. 
%
% To prove the other direction, we assume there exists a non-trivial desingularizer $A$ and derive contradiction. Let $A=\sum_{i=0}^k b_i\tau^i$ where $b_i\in F(x)$. Using division with remainder, we write
% $$A=\sum_{i=0}^k q_i\tau^i +\sum_{i=0}^k\frac{r_i}{d_i}\tau^i,$$
% where $q_i,r_i,d_i\in F[x]$ and $\deg(r_i)<\deg(d_i)$. Let $A'=\sum_{i=0}^k \frac{r_i}{d_i}\tau^i$. Then
% $$A'L=AL-(\sum_{i=0}^k q_i\tau^i)L\in \mathcal{P}.$$
% Since coefficients of $A'L$ are all proper fractions, we have $ \deg(\lcs(A'L))<\deg(\lcs(L))$. Since $\lcs(L)\mid \lcs(A'L)$ by definition, it contradicts to the assumption that $\lcs(L)=\lcs(L)$.
\end{proof}

% \subsection{Special case, Gaussian operators}
% Merge this section with the previous one.

\hyperref[lemma:pcurvprop]{Lemma~\ref*{lemma:pcurvprop}(v)} says  that $\tilde{\chi}(L) = \mathcal{N}(\lcs(L)) \chi(L) \in F[Z][T]$ when $L \in \mathcal{P}$, in other words
\begin{equation} \denom( \chi(L) ) \mid  \mathcal{N}(\lcs(L)). \label{denombound} \end{equation}
Here we give a sharper denominator bound.
\begin{lemma}\label{lemma:denomboundedbytruesing}
For $L\in F[x][\tau]$, $\denom(\mathrm{\chi}(L))\mid \mathcal{N}(\LCs_\infty(L)).$
\end{lemma}
\begin{proof}
%There exists a desingularizer $A\in \mathcal{D}$ such that  $\lcs(AL)=\lcs(L)$\todo{cite}, which leads to the relation
Let $A\in \mathcal{D}$ be an optimal desingularizer of $L$, then $\lcs(AL)=\LCs_\infty(L)$.
%and hence $\mathcal{N}(\lcs(AL))=\mathcal{N}(\LCs_\infty(L)).$
%
%
% \hyperref[lemma:pcurvprop]{Lemma~\ref*{lemma:pcurvprop}(iii)} says
% $\chi(AL)=\chi(A)\chi(L).$
% Applying Gauss's lemma to $F[Z][T]\subseteq F(Z)[T]$, we have
% $$\mathrm{Cont}(\chi(AL))=\mathrm{Cont}(\chi(A))\mathrm{Cont}(\chi(L)),$$ where contents are taken with respect to $T$.
% Since characteristic polynomials are always monic by definition, their denominators are reciprocals of their contents. 
% Hence
From Lemma~\ref{lemma:easy} and Equation~(\ref{denombound}) applied to~$AL$,
$\denom(\chi(L)) \mid \denom(\chi(AL)) \mid \mathcal{N}(\lcs(AL))=\mathcal{N}(\LCs_\infty(L)).$
\end{proof}
Next we show that our denominator bound is exact for Gaussian operators. The next section
will prove the general case by exploiting the fact that any operator has a Gaussian multiple.																		 
\begin{lemma}\label{lemma:thmholdsforgaussian}
If $L\in F[x][\tau]$ is Gaussian, then
$$\denom(\mathrm{\chi}(L))= \mathcal{N}(\lcs_\infty(L)).$$
\end{lemma}
\begin{proof}
%In this proof, a polynomial in $\Fp[Z][T]$ being primitive means it is primitive as a polynomial in $T$ with coefficients in $\Fp[Z]$ (the $\gcd$ of the coefficients is $1$).
Denote $f=\mathrm{Prim}(\chi(L))\in F[Z][T]$.
By \hyperref[lemma:pcurvprop]{Lemma~\ref*{lemma:pcurvprop}(ii)}, $f\in \mathcal{D}L$ so %$\mathrm{Prim}(\chi(L))\in \mathcal{D}L\cap \mathcal{P}$. %F[Z][T]\subseteq F[x][\tau]$ is a rational multiple of $L$,
there exists $Q\in \mathcal{D}$ such that 
\begin{equation}\label{eq:QL}
    QL = f.
\end{equation}
In fact $Q\in\mathcal{P}$ since $L$ is Gaussian.  \hyperref[lemma:pcurvprop]{Lemma~\ref*{lemma:pcurvprop}(v)} says $\tilde{\chi}(Q),\tilde{\chi}(L)\in F[Z][T]$.
%$\mathrm{Prim}(\chi(L))= Q L$.
Applying $\tilde{\chi}$ to Equation~(\ref{eq:QL}), and \hyperref[lemma:pcurvprop]{Lemma~\ref*{lemma:pcurvprop}(vi)}, gives
\begin{equation}\label{eq:chiQL}
    \tilde{\chi}(Q)\tilde{\chi}(L) = \tilde{\chi}(f)=f^p.
\end{equation}
Now $f^p \in F[Z][T]$ is primitive since $f$ is primitive. Then Gauss's lemma implies $\tilde{\chi}(L) \in F[Z][T]$ is primitive. % and thus equals $f$. 
%Since $L$ is Gaussian, $Q\in \mathcal{P}$. and then by \hyperref[lemma:pcurvprop]{Lemma~\ref*{lemma:pcurvprop}(v)}, $\tilde{\chi}(Q),\tilde{\chi}(L)\in F[Z][T]$. It follows from \hyperref[lemma:pcurvprop]{Lemma~\ref*{lemma:pcurvprop}(iii)} that
%$\tilde{\chi}(\mathrm{Prim}(\chi(L)))=\tilde{\chi}(Q)\tilde{\chi}(L)\in F[Z][T]$. 
%Applying \hyperref[lemma:pcurvprop]{Lemma~\ref*{lemma:pcurvprop}(vi)} to $\mathrm{Prim}(\chi(L))\in F[Z][T]$, we see that
%$\tilde{\chi}(\mathrm{Prim}(\chi(L)))=(\mathrm{Prim}(\chi(L)))^p,$
%which is primitive with respect to $T$ by Gauss's lemma. Another implication of Gauss's lemma is that $\tilde{\chi}(L)$ is also primitive. 
It follows that
$$\denom(\chi(L))=\lcs(\tilde{\chi}(L))=\mathcal{N}(\lcs(L))=\mathcal{N}(\LCs_\infty(L))$$
where the last equality comes from Lemma~\ref{lemma:gaussian}, part 4.
\end{proof}
\subsection{Proof for the general case}
\begin{lemma}\label{lemma:lcdenombound}
Suppose $L\in\mathcal{P}$ and $A=\sum_{i=0}^k \frac{n_i}{d_i}\in \mathcal{D}$ is a desingularizer of $L$, where $\frac{n_i}{d_i}\in F(x)$ is in lowest terms for each $i$. Then $\mathcal{N}(d_k)\mid \mathcal{N}(\mathfrak{rp}_\infty(L))$.
\end{lemma}
\begin{proof}
Let $n$ be the order of $L$. The definition of $\lcs_\infty$ and the product of the leading terms of $A$ and $L$ gives
$$\lcs_\infty(L) \mid \tau^{-k-n}(\frac{n_k}{d_k})\lcs(L)$$
and hence $\tau^{-k-n}(\frac{n_k}{d_k})\mathfrak{rp}_\infty(L)\in F[x]$.
Since $\frac{n_k}{d_k}$ is a reduced fraction, we have $\tau^{-k-n}(d_k)\mid \mathfrak{rp}_\infty(L)$, which leads to
$$\mathcal{N}(d_k)=\mathcal{N}(\tau^{-k-n}(d_k))\mid \mathcal{N}(\mathfrak{rp}_\infty(L)).$$
\end{proof}

\begin{lemma}\label{lemma:chidesingbound}
Suppose $L\in\mathcal{P}$ and $A\in \mathcal{D}$ is an optimal desingularizer of $L$. Then there exists a positive integer $N$ such that 
$$\denom(\chi(A))\mid (\mathfrak{rp}_\infty(L))^N.$$
\end{lemma}
\begin{proof}

Write $A=\sum_{i=0}^k \frac{n_i}{d_i}\tau^i$, where $\frac{n_i}{d_i}\in F(x)$ is a reduced fraction for each $i$. We deduce $n_k=1,d_k=\tau^{n+k}(\mathfrak{rp}_\infty(L))$ from the fact that $\lcs(AL)=\lcs_\infty(L)$.
Clearly $d_0d_1\cdots d_kA\in\mathcal{P}$. %Lemma~\ref{lemma:desingularizerdenombound} implies $\mathcal{N}(d_i)\mid \mathcal{N}(\mathfrak{rp}_\infty(L))^{k-j}$.
Then by Equation~(\ref{denombound})			% \hyperref[lemma:pcurvprop]{Lemma~\ref*{lemma:pcurvprop}(v)}, %there exists a sufficiently large $N$ such that
%$$\denom(\chi(A))\mid \mathcal{N}(\lc(d_0d_1\cdots d_kA))$$
%$$\mbox{}\hspace{3cm}
%=\mathcal{N}(d_0d_1\cdots d_{k-1}n_k)\mid \mathcal{N}(\mathfrak{rp}_\infty(L))^N \mathcal{N}(n_k). $$
$$\denom(\chi(A))\mid \mathcal{N}(\lcs(d_0d_1\cdots d_kA))    =\mathcal{N}(d_0d_1\cdots d_{k-1}). $$
%$$\mbox{}\hspace{2.2cm} =\mathcal{N}(d_0d_1\cdots d_{k-1}). $$
Now we bound $\mathcal{N}(d_i)$ in terms of $\mathfrak{rp}_\infty(L)$. Let 
$$A_j=d_kd_{k-1}\cdots d_{j+1}(\sum_{i=0}^j \frac{n_i}{d_i}\tau^i)$$
for $j=0,1,\ldots,k-1$. Notice that
$$A_j-d_kd_{k-1}\cdots d_{j+1}A=-d_kd_{k-1}\cdots d_{j+1}(\sum_{i=j+1}^k \frac{n_i}{d_i}\tau^i)\in\mathcal{P}.$$
This implies $A_jL\in\mathcal{P}$, or equivalently, $A_j$ is a desingularizer of $L$.
Apply Lemma~\ref{lemma:lcdenombound} to $A_j$:
$$\mathcal{N}(\denom(d_kd_{k-1}\cdots d_{j+1}\frac{n_j}{d_j}))\mid \mathcal{N}(\mathfrak{rp}_\infty(L)).$$
Notice that
$$d_j\mid d_kd_{k-1}\cdots d_{j+1}\cdot \denom(d_kd_{k-1}\cdots d_{j+1}\frac{n_j}{d_j}).$$
Therefore
$$\mathcal{N}(d_j)\mid \mathcal{N}(d_kd_{k-1}\cdots d_{j+1})\cdot \mathcal{N}(\mathfrak{rp}_\infty(L)).$$
Recall that $\mathcal{N}(d_k) \mid \mathcal{N}(\mathfrak{rp}_\infty(L))$. By downward induction on $j$, we conclude that
$$\mathcal{N}(d_j)\mid (\mathcal{N}(\mathfrak{rp}_\infty(L)))^{k-j+1}.$$

\end{proof}

We are now ready to finish the proof for \autoref{lemma:denom_pcurv}. 
\begin{proof}[Proof of \autoref{lemma:denom_pcurv}]
It remains to show that
$ \mathcal{N}(\LCs_\infty(L))\mid \denom(\mathrm{\chi}(L))$ for any $L\in \mathcal{D}$.
There exists a sufficiently large $k$ such that $\LCs_k(L)=\lcs_\infty(L)$. Introduce new constants $c_0,\ldots, c_k$ that are algebraically independent over $F$ and denote $E=F(c_0,c_1,\ldots,c_k)$. Let
$$L'=\mathrm{Prim}(\mathrm{LCLM}(c_k\tau^k+\cdots +c_0,L))\in E[x][\tau].$$
\autoref{thm:lclmdesing} says
$\lcs(L')=\lcs_\infty(L)f,$
where $f\in E[x]$ has no non-trivial factor in $F[x]$. It follows from Definition~\ref{def:essentialpart} (see also Equation~\ref{eq:lcideal}) that $\lcs_\infty(L)\mid \lcs_\infty(L')$ and hence $\mathfrak{rp}_\infty(L')\mid f$.
Remark~\ref{rem:lcfieldext} guarantees $\lcs_\infty(L)$ does not change as we shift from $F$ to $E$.
Let $A$ be an optimal desingularizer of $L'$. 
By Lemma~\ref{lemma:pcurvprop} ((iii) and (iv)), we have
\begin{equation}\label{eq:chiAL'}
    \chi(AL')=\chi(A)\chi(c_k\tau^n+\cdots+c_0)\chi(L).
\end{equation}
% Since $c_k$ is a unit in $E[x]$,
\hyperref[lemma:pcurvprop]{Lemma~\ref*{lemma:pcurvprop}(vi)} implies $$\denom(\chi(c_k\tau^n+\cdots+c_0))=1.$$
Applying Lemma~\ref{lemma:easy} to Equation~(\ref{eq:chiAL'}) gives
$$\denom(\chi(AL'))=\denom(\chi(A))\denom(\chi(L)).$$
% (For monic polynomials, $\denom$ is the reciprocal of $\mathrm{Cont}$, which is multiplicative by Gauss's lemma.)
% We remind readers that $\denom$ is distributive here because $\chi(\cdot)$ is monic.
%
%By Lemma~\ref{lemma:chidesingbound}, $\denom(\chi(A))\mid f^N$ for a sufficiently large integer $N$. By 
Since $AL'$ is Gaussian, we know from Lemma~\ref{lemma:thmholdsforgaussian} 
$$\denom(\chi(AL'))=\mathcal{N}(\lcs_\infty(AL')),$$
which equals $\mathcal{N}(\lcs_\infty(L'))$ since $A$ is an optimal desingularizer of $L'$.
As a consequence,
%\begin{equation}\label{eq:denompcurv}
%\mathcal{N}(\lcs_\infty(L))\mid \mathcal{N}(\lcs_\infty(L'))=\denom(\chi(AL'))\mid \mathcal{N}(\lcs(L'))=\mathcal{N}(\lcs_\infty(L))\mathcal{N}(f).
%\end{equation}
%Put \autoref{eq:chiAL'} and \ref{eq:denompcurv} together:
\begin{equation}\label{eq:divisioninE}
\mathcal{N}(\lcs_\infty(L))\mid \mathcal{N}(\lcs_\infty(L'))=\denom(\chi(A))\denom(\chi(L)).
\end{equation}
Lemma~\ref{lemma:chidesingbound} says $\denom(\chi(A))$ is a factor of $f^N$ for some sufficiently large $N$, so $\denom(\chi(A))$ has no non-trivial factor in $F[x]$. 
By taking only factors in $F[x]$ in Equation~\ref{eq:divisioninE}, we obtain the desired result
$\mathcal{N}(\lcs_\infty(L)) \mid \denom(\chi(L)).$

\end{proof}

{%\color{blue}
%\begin{lemma}
%Let $\alpha=\frac{\lc(L)}{\LC(L)}$. Claim that
%$\mathcal(\alpha)^{-1} \tilde{\chi}(L)\in\Fp[Z][T]$.
%\end{lemma}\todo{Same lemma in a different form. The formulation is closer to how it is used in the algorithm.}
%\begin{conjecture}
%$\mathcal{N}(\alpha)^{-1}\tilde{\chi}(L)$ is primitive.
%\end{conjecture}}
{%\color{blue}
\section{Application to Computations}
In this section $F=\Fp$.
\subsection{Algorithm}

Let $L\in \mathcal{P}$ and $\alpha=\mathfrak{rp}_k(L)$.
\autoref{lemma:denom_pcurv} implies that $\mathcal{N}(\alpha)$, which is in $\Fp[Z]$, is a factor of $\Tilde{\chi}(L)$.
Dividing this factor away reduces the degree bound from \hyperref[lemma:pcurvprop]{Lemma~\ref*{lemma:pcurvprop}(v)} to
\begin{equation}\label{eq:pcurvdegbound}
    \deg_{Z}(\,\mathcal{N}(\alpha)^{-1}\tilde{\chi}(L)\,)\leq \deg_x(L)-\deg_x(\alpha)
\end{equation}
which becomes an equality when $k$ is sufficiently large. However, we use $k=1$ to minimize the time spent computing $\alpha$.
The reduced degree bound % from knowing an a priori factor $\mathcal{N}(\alpha)$ of $\tilde{\chi}(L)$
allows us to recover $\chi(L)$ from a lower precision $Z$-adic expansion.
That leads to the following algorithm.
%Inequality~\ref{eq:pcurvdegbound} infers that instead of computing $\tilde{\chi}_p(L)$ up to the precision $O(Z^{\deg_x(L)+1})$, it suffices to compute $\mathcal{N}(\alpha)^{-1}\tilde{\chi}_p(L)$ up to $O(Z^{\mathrm{prec}})$, where $\mathrm{prec}=\deg_x(L)-\deg_x(\alpha)+1$.
%Suppose $\mathcal{N}(\alpha)=Z^v\beta$, where $v\in\mathbb{N}$ and $Z\nmid \beta$. Then $\beta$ is invertible in $\Fp[[Z]]$ and
%$$\mathcal{N}(\alpha)^{-1}\tilde{\chi}_p(L)=(\beta^{-1})(Z^{-v}\tilde{\chi}_p(L)).$$
%Therefore it suffices to compute $\beta^{-1}$ up to $O(Z^{\mathrm{prec}})$ and $\tilde{\chi}_p(L)$ up to $O(Z^{\mathrm{prec}+v})$. 
%Considering the cost-benefit balance {\color{blue}(\autoref{rem:balance})}, setting $\alpha=\frac{\lc(L)}{\LC_1(L)}$ is a good choice in practice.
%We then have the following algorithm.
}

\begin{algorithm}
\SetKwInOut{Input}{Input}\SetKwInOut{Output}{Output}
\caption{\texttt{Xi\_p\_desing}}\label{alg:xipdesing}
\Input{prime $p$ and $L\in \Fp[x][\tau]$}
\Output{$\mathrm{Prim}(\chi(L))\in \Fp[Z][T]$}
\BlankLine

     Pick $k \geq 1$ and compute $\LCs_k(L)$ and $\alpha:=\mathfrak{rp}_k(L)\in \Fp[x]$. We use $k=1$ to minimize the time spent in this step.
     
    Compute $\mathcal{N}(\alpha)\in\Fp[Z]$. Let $v$ be its $Z$-adic valuation in $\Fp[[Z]]$ and let $\beta=Z^{-v}\mathcal{N}(\alpha)\in \Fp[Z]$.\newline
    For computing $\mathcal{N}(\cdot)$ see Step 3 of \texttt{Xi\_theta\_d} in \cite{Bostan:2014:FAC:2608628.2608650}.
    
     Let $d_1=\deg_Z(\beta)$ and $d:=\deg_x(L)$. %Compute \texttt{Xi\_theta\_d}$(L)$ up to the precision $O(Z^{d-d_1+1})$. Denote the result by $\chi_1$. Then $\tilde{\chi}(L)=\chi_1+O(Z^{d-d_1+1})$.\newline
    Apply the BCS algorithm with $d$ replaced by $d-d_1$ to $L$ to obtain $\tilde{\chi}(L)$ up to the precision $O(Z^{d-d_1+1})$. Denote the result by $\chi_1$.
    %{Remark: {Changing the precision of the BCS algorithm can be done easily by replacing $d$ by $d-d_1$ in the original algorithm.}}
    
     Compute $\beta^{-1}$ in $\Fp[[Z]]$ up to $O(Z^{d-d_1-v+1})$.\newline
    {{This can be done by applying the extended Euclidean algorithm to $\beta$ and $Z^{d-d_1-v+1}$.}}
    
     Compute  $\beta^{-1}\cdot (Z^{-v}\chi_1)$ in $\Fp[[Z]][T]$ up to the precision $O(Z^{d-d_1-v+1})$. This gives
    $\mathcal{N}(\alpha)^{-1}\chi_1\in\Fp[Z][T]$. \newline Return its primitive part (with respect to $T$).\newline
    Note: $\mathcal{N}(\alpha)^{-1}\chi_1$ and $\mathcal{N}(\alpha)^{-1}\tilde{\chi}(L)$ agree to precision $O(Z^{d-d_1-v+1})$ which suffices % since $d-d_1-v$ is a $Z$-degree bound for $\mathcal{N}(\alpha)^{-1}\tilde{\chi}(L)$ (Inequality
    by~(\ref{eq:pcurvdegbound}).

\end{algorithm}
%\begin{remark}
Step 3 is where we save CPU time over the original algorithm from \cite{Bostan:2014:FAC:2608628.2608650} if $d_1>0$. If $d_1=0$ then there is no improvement in efficiency. However, as we will see in the following section, the extra steps cost very little time.
%\end{remark}

\iffalse
\todo[inline]{The above doesn't explain in details how to compute $\beta^{-1}\chi_1$. Not sure which is better.}

\begin{itemize}
    \item[1] Compute $\LC_1(L)$ and $\alpha:=\frac{\lc(L)}{\LC_1(L)}\in \Fp[x]$.
    \item[2] Compute $\beta:=\mathcal{N}(\alpha)\in\Fp[Z]$.\newline
    Remark: See Step 3 of \texttt{Xi\_theta\_d} in \cite{Bostan:2014:FAC:2608628.2608650} regarding how to compute the factorial of a polynomial.
    \item[3] Let $v_1$ be the valuation of $\beta$ in $\Fp[[Z]]$ and $d_1=\deg_Z(\beta)$. Denote $\beta^*:=Z^{-v_1}\beta$. Then $\beta^*$ is invertible in $\Fp[[Z]]$. 
    \item[4] Compute $(\beta^*)^{-1}$ up to precision $O(Z^{d-v_1+1})$.\newline
    Remark: This step can be done by applying extended Euclidean algorithm to $\beta^*$ and $Z^{d-v_1+1}$.
    \item[5] Set $d:=\deg_x(L)$ and compute \texttt{Xi\_p}$(L)$ up to the precision $O(Z^{d-d_1+v_1+1})$. Denote the result by $\chi_1$. 
    \item[6] Compute $\beta^{-1}\chi_1=(\beta^*)^{-1}(Z^{-v_1}\chi_1)$ in $\Fp[[Z]][T]$ up to the precision $O(Z^{d-d_1+1})$ and return.\newline
    Remark: \autoref{lemma:denom_pcurv} implies $Z^{-v_1}\chi_1\in\Fp[[Z]]$ and it is known at precision $O(Z^d-d_1+1)$.
\end{itemize}
\fi

\subsection{Implementation and timings}\label{sec:pcurv_timing}

{
Our Magma implementation of algorithm \texttt{Xi\_p\_desing} is available at \url{https://www.math.fsu.edu/~yzhou/magma/}, together with experiments on a variety of operators.
One should load the implementation of \cite{Bostan:2014:FAC:2608628.2608650} at \url{https://github.com/schost/pCurvature} (file \texttt{pCurvature.mgm}) prior to ours. 
%We experiment our implementation on a variety of operators and compare timings with the implementation of BCS algorithm. %The machine used for the test is \todo{add info about the computer, or maybe just delete it if not necessary}. 
%The complete test can be found on our website. 

In the following table the data for two operators from OEIS (\cite{A002777}, \cite{A151329}) is presented. 
%Operators that we used for the test can be classified into four categories: the order 46 operator (\todo{cite}), operators from OEIS, operators constructed by taking LCLM, random operators.
Here $d_1$ is defined in the Step 3 of \texttt{Xi\_p\_desing}; each running time is the average of ten runs.%; $t_1$ and $t_2$ are running times of original BCS algorithm and the new one with desingularization, respectively.

\begin{table}[!h]
    \centering
%\begin{tabular}{ c|p{1.5cm}|p{1.5cm}|p{1.5cm}|p{2cm}|p{1.5cm}|p{2.3cm}|  }
\begin{tabular}{ |c|c|c|c|c|c|  }
 \hline
 OEIS index & order & $x$-degree & $d_1$ & BCS & \texttt{Xi\_p\_desing} \\
 \hline
 A151329 & 9 & 18 & 10 & 17.2s & 9.6s \\
 \hline
 A002777 & 4 & 3 & 0 & 6.97s & 7.01s \\
 \hline 
 %%165544 & 4 & 5 & 0.7s & 2 & 0.84s \\
 %%\hline 
 %%002777 & 4 & 4 & 0.58s & 2 & 0.47s \\
\end{tabular}
\caption{Timings for operators from OEIS. $p=27457$.}
\label{table:oeis}
\end{table}
For the recurrence for OEIS A002777, we expect $\texttt{Xi\_p\_desing}$ to be slower than BCS since $d_1=0$. However, the running time difference between two algorithms is nearly  unnoticeable. 

We also tested our algorithm on operators that are LCLMs of two operators~\cite{implementations}.
 %Operators in the form $\mathrm{LCLM}(L_1,L_2)$
Such operators tend to have many apparent singularities and hence benefit more from our approach.%; such operators benefit the most from our approach. On the other side, random operators are likely to have no apparent singularity at all; for these operators our algorithm is 

\section{Fast Algorithms for Desingularization at Order 1}\label{sec:desing1}

\subsection{First algorithm}\label{sec:LC1}
In this section we present our first speedup of the order-1 LCLM method. We used it for Step 1 of \autoref{alg:xipdesing}.

The order-1 LCLM method computes $L'=\mathrm{LCLM}(L,\tau-c)$ where $c$ is a new constant (or a random number in the Monte-Carlo version).
To speed this up, our idea is to obtain $\LCs_1(L)$ while only computing a portion of $L'$. %? In the following an affirmative answer is given.

First we express of $L'$ in terms of $c$ and coefficients of $L$.
Suppose $L=\sum_{i=0}^n a_i\tau^i\in F(x)[\tau]$. Then
%We start by computing $\mathrm{LCLM}(L,\tau-c)$ for $L=\sum_{i=0}^n a_i\tau^i$ and $c\in F$. Let
$L'=\sum_{i=0}^n c^i L_i,$
where
\begin{equation}\label{eq:Li}
    L_{i}=a_{i}\tau L-\tau(a_{i-1}) L= (a_i\tau-\tau(a_{i-1}))L,
\end{equation}
where $a_i=0$ for $i<0$ and $i>n$. Clearly this $L'$ is a left multiple of $L$. To verify it is also a left-multiple of $\tau-c$, use the fact that the remainder of $\tau^i$ right-divided by $\tau-c$ is $c^i$. We skip the tedious computation. As a result $L'$ is an LCLM of $L$ and $\tau-c$.
%noticing $\tau^i-c^i$ is a left multiple of $\tau-c$, replace $\tau^i$ by $c^i$ in $L'$ and check the obtained expression is equal to 0. 

The order-1 LCLM method computes $L'$ which amounts to compute all $L_i$'s. The following proposition shows that one can provably obtain $\LCs_1(L)$
from just a subset of the $L_i$'s.

\begin{proposition}\label{prop:LC1} Let $L=\sum_{i=0}^n a_i\tau^i\in F[x][\tau]$. 
Let $L_i$ be defined by \autoref{eq:Li}, where $a_i=0$ for $i<0$ and $i>n$.
If
\begin{equation}\label{eq:gcdcondition}
    \gcd(a_{i_1},a_{i_2},\ldots, a_{i_k})=1,
\end{equation}
 then
$$\gcd(\LCs_0(L_{i_1}),\LCs_0(L_{i_2}),\ldots, \LCs_0(L_{i_k}))=\LCs_1(L).$$
%Then
%$\LCs_0(L_i)\mid \tau^{-n-1}(a_i)\LCs_1(L).$
\iffalse\begin{itemize}
    \item[1] $$\LC_0(L_i)\mid a_i\tau(\LC_1(L));$$
    \item[2] $$\TC_0(L_{i+1})\mid \tau(a_i)\TC_1(L).$$
\end{itemize}\fi
\end{proposition}

\iffalse
\begin{proposition}\label{prop:LC1} Let $L=\sum_{i=0}^n a_i\tau^i\in F[x][\tau]$ be primitive. Due to the primitivity of $L$, there exists $I=\{i_1,i_2,\ldots, i_m\}\subseteq \{0,1,\ldots,n\}$ such that $a_i\neq 0$ for $i\in I$ and
$$\gcd(a_{i_1},a_{i_2},\cdots, a_{i_m})=1.$$
Let $L_i$ be defined by \autoref{eq:Li}, where we assume $a_i=0$ for $i<0$ and $i>n$. Claim that
\begin{itemize}
    \item[1] $$\gcd(\LC_0(L_i):i\in I)=\tau(\LC_1(L));$$
    \item[2] $$\gcd(\mathrm{TC}_0(L_{i+1}):i\in I)=\mathrm{TC}_1(L).$$
\end{itemize}
\end{proposition}
\fi

The proof will be given in the next section. Note that there exist $i_1,i_2,\ldots, i_k$ satisfying the gcd condition (\autoref{eq:gcdcondition}) if and only if $L$ is primitive. The proposition immediately implies \autoref{alg:LC1}.

\begin{algorithm}
\SetKwInOut{Input}{Input}\SetKwInOut{Output}{Output}
\caption{$\LCs_1$}\label{alg:LC1}
\Input{a primitive operator $L=\sum_{i=0}^n a_i \tau^i\in F[x][\tau]$}
\Output{$\LCs_1(L)$}
\BlankLine

 Find $I\subset \{0,1,\ldots, n\}$ such that $a_{i}\neq 0$ for any $i\in I$ and $\gcd(a_i\mid i\in I)=1$. % 
 Note: the algorithm is still correct if we allow $a_i=0$, but that $i$ is redundant since it does not affect the gcd at all.
 
 Compute $L_i$ for $i\in I$ by \autoref{eq:Li}.
 
 Return $\gcd(\LCs_0(L_i): i\in I)$.

\end{algorithm}

 {\iffalse
\begin{algorithm}[H]
\SetKwInOut{Input}{Input}\SetKwInOut{Output}{Output}
\caption{$\LC_1$}\label{alg:LC1}
\Input{a primitive operator $L=\sum_{i=0}^n a_i \tau^i\in F[x][\tau]$}
\Output{$\LC_1(L)$}
\BlankLine

$a,l\leftarrow \tau(a_n),\tau(a_n)$;

$I\leftarrow \{0,1,\ldots, n\}$;

\While{$a\neq 1$}{
Choose $i\in I$ randomly;

$I\leftarrow I-\{i\}$;

    \If{$a_i\neq 0$}{
        $Li\leftarrow (a_i\tau-\tau(a_{i-1}))L$;
        
        $l\leftarrow \gcd(l,\LC_0(Li))$;

        $a\leftarrow \gcd(a,a_i)$;
    }
}
\Return{$\tau^{-1}(l)$};
\end{algorithm}
\fi

\iffalse\begin{algorithm}[H]
\SetKwInOut{Input}{Input}\SetKwInOut{Output}{Output}
\caption{$\mathrm{TC}_1$}\label{alg:TC1}
\Input{a primitive operator $L=\sum_{i=0}^n a_i \tau^i\in F[x][\tau]$}
\Output{$\TC_1(L)$}
\BlankLine

Find $I\subset \{0,1,\ldots, n\}$ such that $a_{i}\neq 0$ for any $i\in I$ and $\gcd(a_i\mid i\in I)=1$.

Compute $L_{i+1}=(a_{i+1}\tau-\tau(a_{i}))L$ for $i\in I$.

Compute $t=\gcd(\mathrm{TC}_0(L_{i+1}): i\in I)$.

return $t$.

\end{algorithm}\fi
}

\begin{remark}
 Computing $\LCs_0(L_i)=\lcs(\mathrm{Prim}(L_i))$ %or $\mathrm{TC}_0(L_i)$ 
 is the most time-consuming part in the algorithm, because $L_i$ has twice the $x$-degree as $L$.
\end{remark}

\iffalse
{\color{blue}
The following is the pseudocode version of \autoref{alg:LC1}.\todo{delete entirely or keep for thesis}

\begin{algorithm}[H]
\SetKwInOut{Input}{Input}\SetKwInOut{Output}{Output}
\caption{$\LC_1$}\label{alg:LC1pseudo}
\Input{a primitive operator $L=\sum_{i=0}^n a_i \tau^i\in F[x][\tau]$}
\Output{$\LC_1(L)$}
\BlankLine

$a,l\leftarrow \tau(a_n),\tau(a_n)$;

$I\leftarrow \{0,1,\ldots, n\}$;

\While{$a\neq 1$}{
Choose $i\in I$ randomly;

$I\leftarrow I-\{i\}$;

    \If{$a_i\neq 0$}{
        $Li\leftarrow (a_i\tau-\tau(a_{i-1}))L$;
        
        $l\leftarrow \gcd(l,\LC_0(Li))$;

        $a\leftarrow \gcd(a,a_i)$;
    }
}
\Return{$\tau^{-1}(l)$};
\end{algorithm}
}
\fi

\subsection{Proof}
    %In this section part 1 of \autoref{prop:LC1} is fully proved and the details of the proof for part 2 are omitted (the reason is explained in \autoref{rem:LCTC}). 
    Always assume $L=\sum_{i=0}^n a_i\tau^i\in F[x][\tau]$ is primitive and $L_i$ is defined by \autoref{eq:Li}.
\begin{lemma}\label{lemma:desing_poly}
There exists $b\in F[x]$ such that 
$$\mathrm{Cont}((\tau-b)L)=\tau^{n+1}(\mathfrak{rp}_1(L)).$$
\iffalse\begin{itemize}
    \item[1] There exists $b\in F[x]$ such that 
$$\LC_0((\tau-b)L)=\tau(\LC_1(L)).$$
    \item[2] There exists $b'\in F[x]$ such that
    $$\LC_0((b'\tau-1)L)=\mathrm{TC}_1(L).$$
\end{itemize}\fi

\end{lemma}
\begin{proof}
Let $A\in\mathcal{D}$ be an optimal desingularizer of $L$ at order $1$. Then $A=\frac{1}{d_1}\tau-\frac{n_2}{d_2}$, where $d_1=\tau^{n+1}(\mathfrak{rp}_1(L))$ and $\frac{n_2}{d_2}\in F(x)$ is a reduced fraction. Let $b=d_1\frac{n_2}{d_2}$. Observe that
\begin{equation}\label{eq:lc1b1}
    bL=\tau\cdot L- d_1AL\in\mathcal{P}.
\end{equation}
    
Due to $L$ being primitive, $b$ has to be a polynomial.
Since $A$ is an optimal desingularizer of $L$ at order 1, $AL\in\mathcal{P}$ is primitive; otherwise dividing out the content of $AL$ yields a more optimal desingularizer. By rearranging Equation~\ref{eq:lc1b1} we see that $\frac{1}{d_1}(\tau-b)L=AL$ is primitive,
which completes the proof. 
\end{proof}

\iffalse\begin{proof}
Suppose $b_1\tau-b_2\in \mathrm{Des}$ such that $\lc((b_1\tau-b_2)L)=\tau(\LC_1(L))$. The existence of such an operator is guaranteed by the fact that $F[x]$ is a PID. % Namely $(b_1\tau-b_2)L\in F[x][\tau]$ and its leading coefficient is $\tau(\LC_1(L))$. 
Then
$$\frac{1}{b_1}=\tau(\frac{\lc(L)}{\LC_1(L)})\in F[x].$$
Therefore,
$$\frac{b_2}{b_1}L=\tau L-\frac{1}{b_1}(b_1\tau-b_2)L\in F[x][\tau].$$
Because of the primitiviness of $L$, $\frac{b_2}{b_1}\in F[x]$. It is easy to check that
$\LC_0((\tau-\frac{b_2}{b_1})L)=\tau(\LC_1(L))$.
\end{proof} \fi

\iffalse
\begin{lemma}
Let $b,b'$ be the same as in Lemma~\ref{lemma:desing_poly}. Claim that
\begin{itemize}
    \item[1] $$\gcd(\tau(a_{j-1})-ba_{j}\mid j=0,1,\ldots,n+1)=\tau(\frac{\lc(L)}{\LC_1(L)});$$
    \item[2] $$\gcd(b'\tau(a_{j-1})-a_{j}\mid j=0,1,\ldots,n+1)=\frac{\tc(L)}{\TC_1(L)}.$$
\end{itemize}
\end{lemma}
\begin{proof}
\end{proof}
\fi

\begin{theorem}\label{thm:LC1}
 Let $C=(c_0,c_1,\ldots, c_{n+1})\in F^{n+2}$. Denote
%$C_1=%\sum_{i=0}^{n} c_i a_i=
%\sum_{i=0}^{n+1} c_i a_i$, $C_0=\sum_{i=0}^{n+1} c_i \tau(a_{i-1})$ and $L'=(C_1\tau-C_0)L$.
$$C_1=%\sum_{i=0}^{n} c_i a_i=
\sum_{i=0}^{n+1} c_i a_i, \quad C_0=\sum_{i=0}^{n+1} c_i \tau(a_{i-1}), \quad L'=(C_1\tau-C_0)L.$$
Then
    $$\LCs_1(L)\mid \LCs_0(L')\mid \tau^{-n-1}(C_1)\LCs_1(L).$$
\iffalse\begin{itemize}
    \item[1] If $C_1\neq 0$, then
    $$\LC_0(L')\mid C_1\tau(\LC_1(L)).$$
    \item[2] If $C_0\neq 0$, then
    $$\TC_0(L')\mid C_0\TC_1(L).$$
\end{itemize}\fi
%Claim that:
%\begin{itemize}
%    \item[1] If 
%$$\gcd(a_n, C_1)=1$$%\sum_{i=0}^{n} r_i a_i)=1$$
% then
%$$\gcd(a_n,\tau^{-1}(\LC_0(L')))=\LC_1(L);$$
%    \item[2] if 
%    $$\gcd(a_0,C_0)=1$$
%    then
%    $$\gcd(a_0,\mathrm{TC}_0(L'))=\mathrm{TC}_1(L).$$
%\end{itemize}
\end{theorem}
\begin{proof}
Assume $C_1\neq 0$ since otherwise it is trivial.

The relation $\LCs_1(L)\mid \LCs_0(L')$ immediately follows from the definition of $\lcs_1$.																							
By Lemma~\ref{lemma:desing_poly}, there exists $b\in F[x]$ such that $\mathrm{Cont}((\tau-b)L)=\tau^{n+1}(\mathfrak{rp}_1(L))$.
Since
$$(\tau-b)L=\sum_{i=0}^{n+1} (\tau(a_{j-1})-ba_{j})\tau^j,$$
we have
\begin{equation}\label{eq:coeff_appsing}
   \gcd(\tau(a_{j-1})-ba_{j}\mid j=0,1,\ldots,n+1)=\tau^{n+1}(\mathfrak{rp}_1(L)).
\end{equation}
Notice that
$$L'=C_1(\tau-b)L+(bC_1-C_0)L,$$
and in particular
$$bC_1-C_0=\sum_{i=0}^{n+1}bc_ia_i-\sum_{i=0}^{n+1}c_i\tau(a_{i-1})=\sum_{i=0}^{n+1}c_i(ba_i-\tau(a_{i-1}))$$
is a multiple of $\tau^{n+1}(\mathfrak{rp}_1)$ due to \autoref{eq:coeff_appsing}. Hence
$\frac{1}{\tau^{n+1}(\mathfrak{rp}_1)}L'\in F[x][\tau].$
When $C_1\neq 0$, %$\lc(L')=C_1\tau(\lc(L))$. 
$$\lcs(\frac{1}{\tau^{n+1}(\mathfrak{rp}_1(L))}L')=\frac{1}{\mathfrak{rp}_1(L)}\tau^{-n-1}(C_1)\lcs(L)=\tau^{-n-1}(C_1)\lcs_1(L).$$
Then we have
$$\LCs_0(L')=\lcs(\mathrm{Prim}(L'))\mid \lcs(\frac{1}{\tau^{n+1}(\mathfrak{rp}_1(L))}L')=\tau^{-n-1}(C_1)\lcs_1(L).$$
%$$L'=a_i(\tau-b)L+(ba_i-\tau(a_{i-1}))L$$
%and $\frac{a_i(\tau-b)L}{\tau(\frac{\lc(L)}{\LC_1(L)})}\in F[x][\tau],\frac{ba_i-\tau(a_{i-1})}{\tau(\frac{\lc(L)}{\LC_1(L)})}\in F[x]$. Then
%$\frac{L_i}{\tau(\frac{\lc(L)}{\LC_1(L)})}\in F[x][\tau]$. As a result, when $a_i\neq 0$,
%$$\LC_0(L_i)\mid \frac{\lc(L_i)}{\tau(\frac{\lc(L)}{\LC_1(L)})}=a_i\tau(\LC_1(L)).$$
%On the other hand, since $L_i=(a_i\tau-\tau(a_{i-1}))L$, when $a_i\neq 0$, by the definition of $\LC_1(L)$, we have
%$$\tau(\LC_1(L))\mid \LC_0(L_i).$$
\end{proof}
\begin{proof}[Proof of Proposition~\ref{prop:LC1}]
In \autoref{thm:LC1}, setting $c_i=1$ for some $i$ and $c_j=0$
for any $j\neq i$ yields
$$\lcs_1(L)\mid \lcs_0(L_i)\mid \tau^{-n-1}(a_i)\lcs_0(L).$$
The desired result follows immediately.						
\end{proof}

\subsection{Desingularizing both leading and trailing coefficients}\label{sec:LC1speedup}

The variation in this section handles both leading and trailing singularities.
It uses only one $L_i$ (defined in \autoref{eq:Li}) without checking the $\gcd$ condition (\autoref{eq:gcdcondition}), since
most apparant singularities are already detected with one $L_i$.

{%\color{blue}
In the algorithm, $\mathfrak{tc}_1$ denotes the \emph{essential part of trailing coefficient at order 1}, which is the counterpart of $\LCs_1$ for trailing coefficients.
}

\begin{algorithm}
\SetKwInOut{Input}{Input}\SetKwInOut{Output}{Output}
\caption{$\lcs$1\_$\mathfrak{tc}$1}\label{alg:speedup}
\Input{a primitive operator $L=\sum_{i=0}^n a_i \tau^i\in F[x][\tau]$ with $a_0a_n\neq 0$}
\Output{$l,t\in F[x]$ such that $\LCs_1(L)\mid l \mid \lcs(L)$ and $\mathfrak{tc}_1(L)\mid l \mid \mathfrak{tc}(L)$}
\BlankLine

$i\leftarrow \lfloor \frac{n}{2}\rfloor$

% \While{$a_i=0$ or $a_{i-1}= 0$}{
% $i\leftarrow i+1$
% }
$L_i\leftarrow (a_i\tau-\tau(a_{i-1}))L$

$l,t\leftarrow \LCs_0(L_i),\TC_0(L_i)$

$l,t\leftarrow \gcd(\tau^{-n}(a_n),l),\gcd(a_0,t)$

\Return{$l,t$}
\end{algorithm}

\subsection{Examples and comparisons}\label{sec:desingtiming}

We have implemented \autoref{alg:LC1} and \autoref{alg:speedup} in Maple and SageMath, and done some experiments to compare the running time of our algorithm with the order-1 LCLM method. All can be found at \cite{implementations}. Below we give an experiment we did in Maple. 
\begin{example}
In this example the base field is $\mathbb{Q}$. We took random operators
$$L_1=(26x^4+20)\tau^{11}-96x^3\tau^9+64x^5\tau^8+45x^{11}\tau^4-x^2\tau^3,$$
$$L_2=-55x^3\tau^7+85x^3\tau^4+64x^4\tau^3+(-14x^8-20x^4)\tau+79x,$$
and then computed
$$L=\mathrm{Prim}(\mathrm{LCLM}(L_1,L_2)).$$
The $x$-degree of $L$ is $109$.
We desingularize $L$ using three different algorithms. For the LCLM method we used the Monte-Carlo version and randomly choose $c=7$.
The results are shown in the \autoref{table:desing}, where each time is the average of ten runs. 

\begin{table}
    \centering
\begin{tabular}{ |c|c|c|  }
 \hline
 algorithms & running time & $x$-degree in output  \\ 
 \hline
 Order-1 LCLM &   1.191s  & 6    \\
 \hline
 $\LC1$ & 0.055s & 6 \\
 \hline
 $\lcs$1\_$\mathfrak{tc}$1 & 0.092s & 6 \\
 \hline
\end{tabular}
\caption{Comparison of different desingularization algorithms}
\label{table:desing}
\end{table}
One might expect $\lcs$1\_$\mathfrak{tc}$1 to be slower than $\LC1$ because it treats both the leading and trailing coefficient, however, we expected
it to be faster because it corresponds to taking just one $L_i$ in $\LC1$. % It remains to be seen why this was not the case.
% Perhaps an implementation issue, or perhaps $\LC1$ also uses few $L_i$'s  --->  GOOD POINT, but why didn't we check it?  So lets not write it for now.
%

% TEMPORARILY DELETED BECAUSE THE CHANGES COULDN'T BE COMBINED.
%
% Finally we note that the Monte-Carlo version of LCLM and $\lcs$1\_$\mathfrak{tc}$1 are not guaranteed to always find $\LCs_1(L)$; they may return a multiple of $\lcs_1(L)$.
% In this example, the last column of \autoref{table:desing} indicates they did find $\lcs_1(L)$.

\end{example}

\section{Future Work}
\subsection{Application to Pag\`{e}s' algorithm}

Pag\`{e}s' Algorithm computes $\chi_p(L)$ for $L\in\mathbb{Z}[x][\tau]$ with $\lc(L) \in \mathbb{Z}$, but with minor adjustments it applies to all recurrence operators in $\mathbb{Z}[x][\tau]$.
We expect desingularization to be beneficial here as well.
% , since it is based on the BCS algorithm.

\subsection{Differential case}
The desingularization improvement should also work for the differential case or Ore operators. For a differential operator $L=\sum_{i=0}^n a_i\partial^i\in F[x][\partial]$,
we can write
$\mathrm{LCLM}(L,\tau-c)=\sum_{i=0}^n c^i L_i$, where
$$L_i=(a_i\partial -(a_{i-1}+a_i'))\,L.$$
% We conjecture that in order to compute $\lcs_1(L)$, it suffices to compute some $L_i$'s instead of all.
%

We expect that there should also be a differential analog of our main result, \autoref{lemma:denom_pcurv}.

 \bibliographystyle{plain}
 \bibliography{myrefs}

\end{document}